\documentclass[11pt]{article}
\pdfoutput=1

\usepackage{fullpage}

\usepackage[T1]{fontenc}
\usepackage[utf8]{inputenc}
\usepackage{cite}
\usepackage{amsmath}
\usepackage{subfigure}
\usepackage{pgfplots}
\usepackage{xspace}
\usepackage{booktabs} % good-looking tables
\usepackage[activate={true,nocompatibility},final,tracking=true,kerning=true,spacing=true]{microtype}
\microtypecontext{spacing=nonfrench}

\pgfplotsset{compat=newest}
\newcommand{\sq}{\hbox{\rlap{$\sqcap$}$\sqcup$}}
\newcommand{\qed}{\hspace*{\fill}\sq}
\newenvironment{proof}[1][]{\noindent {\textbf{Proof#1.}}\ }{\qed\par\vskip 4mm\par}
\newdimen\endofsize\endofsize=0.5em

\newtheorem{theorem}{Theorem}
\newtheorem{lemma}[theorem]{Lemma}

\newcommand{\card}[1]{\left\vert{#1}\right\vert}

\newcommand{\Oh}[1]{\mathcal{O}\!\left( #1\right)}
\newcommand{\oh}[1]{\mathrm{o}\!\left( #1\right)}
\newcommand{\Ohsmall}[1]{\mathcal{O}(#1)}

\newcommand{\vertex}{v}
\newcommand{\vertexA}{u}
\newcommand{\vertexB}{v}
\newcommand{\vertexC}{w}
\newcommand{\Alphabet}{\Sigma}
\newcommand{\Tree}{T}
\newcommand{\Toptree}{\mathcal{T}}
\newcommand{\TopDAG}{\mathcal{TD}}
\newcommand{\NumN}{n_{\Tree}}
\newcommand{\NumNTD}{n_{\TopDAG}}
\newcommand{\NumS}{\sigma}
\newcommand{\J}{j}
\newcommand{\thresh}{t}
\newcommand{\SubS}[1]{s(#1)}
\newcommand{\Height}[1]{h(#1)}
\newcommand{\SB}{x}

\newenvironment{absolutelynopagebreak}
  {\par\nobreak\vfil\penalty0\vfilneg
   \vtop\bgroup}
  {\par\xdef\tpd{\the\prevdepth}\egroup
   \prevdepth=\tpd}

\usepackage{hyperref}

\hypersetup{
  pdfauthor={Lorenz Hübschle-Schneider and Rajeev Raman},
  pdftitle={Tree Compression with Top Trees Revisited},
  pdfsubject={Tree Compression, Top Trees},
  urlcolor=blue,
}

\begin{document}

\title{Tree Compression with Top Trees Revisited}

\author{Lorenz Hübschle-Schneider \\ \texttt{\href{mailto:huebschle@kit.edu}{huebschle@kit.edu}}\\Institute of Theoretical Informatics\\Karlsruhe Institute of Technology\\Germany \and Rajeev Raman \\ \texttt{\href{mailto:r.raman@leicester.ac.uk}{r.raman@leicester.ac.uk}} \\Department of Computer Science\\University of Leicester\\United Kingdom}

\date{}

\maketitle
\setcounter{footnote}{0}

\begin{abstract}
We revisit tree compression with top trees
(Bille et al.~\cite{TopTrees2013}),
and present several improvements to the compressor and
its analysis. By significantly reducing the amount of information stored and
guiding the compression step using a RePair-inspired heuristic, we obtain a fast
compressor achieving good compression ratios, addressing
an open problem posed by~\cite{TopTrees2013}.
We show how, with relatively small overhead, the
compressed file can be converted into an in-memory representation that
supports basic navigation operations in worst-case logarithmic time without
decompression. We also show a much improved worst-case bound on the size
of the output of top-tree compression (answering an open question
posed in a talk on this algorithm by Weimann in 2012).

%\keywords{Tree Compression · Grammar Compression · Top Trees · XML compression}
\end{abstract}

%=========================================================================
%  Introduction
%=========================================================================
\section{Introduction}\label{s:intro}

Labelled trees are one of the most frequently used nonlinear data structures in
computer science, appearing in the form of suffix trees, XML files, tries, and
dictionaries, to name but a few prominent examples. These trees are
frequently very large, prompting a need for compression for on-disk storage.
Ideally, one would like specialized tree compressors to certainly get
much better compression ratios than general-purpose compressors such as
\texttt{bzip2} or \texttt{gzip}, but also for the compression to
be fast; as Ferragina et al. note~\cite[p4:25]{FerraginaSuccinct2009}.
\footnote{Their remark is about XML tree compressors but applies to
general ones as well.}

In fact, it is also frequently necessary to hold such trees in main memory and
perform complex navigations to query or mine them. However, common in-memory
representations use pointer data structures that have significant
overhead---e.g. for XML files,
standard DOM\footnote{\emph{Document Object Model}, a common interface for
interacting with XML documents} representations are typically~8-16~times larger than
the (already large) XML file~\cite{SiXMLWhitePaper,SpaceEffDOM2007}. To
process such large trees, it is essential to
have compressed in-memory representations that
\emph{directly} support rapid navigation and queries, without partial or
full decompression.

Before we describe previous work, and compare it with
ours, we give some definitions.
A \emph{labelled tree} is an ordered, rooted tree whose nodes have labels from
an alphabet $\Sigma$ of size $\card{\Sigma}=\NumS$. We consider the following
kinds of redundancy in the tree structure. \emph{Subtree repeats} are repeated
occurrences of \emph{rooted subtrees}, i.e.\ a node and all of its descendants,
identical in structure and labels. \emph{Tree pattern repeats} or
\emph{internal repeats} are repeated occurrences of \emph{tree patterns}, i.e.\
connected subgraphs of the tree, identical in structure as well as labels.

%=========================================================================
%  Previous Work
%=========================================================================

\subsection{Previous Work}\label{s:intro:prev}

Nearly all existing compression methods for labelled trees follow one of three
major approaches: \emph{transform-based compressors} that transform the tree's
structure, e.g. into its minimal DAG, \emph{grammar-based compressors} that
compute a tree grammar, and--although not compression--\emph{succinct
representations} of the tree.

\paragraph*{Transform-Based Compressors.} %.....................................
We can replace subtree repeats by edges to a single shared instance of the
subtree and obtain a smaller Directed Acyclic Graph (DAG) representing the tree.
The smallest of these, called the \textit{minimal DAG}, is unique and can be computed in linear
time~\cite{MinDag1980}. Navigation and path queries can be supported in
logarithmic time~\cite{PQXML2003,RandomAccessGrammars2011}. While its size can
be exponentially smaller than the tree, no compression is achieved in the worst
case (a chain of nodes with the same label is its own minimal DAG, even though
it is highly repetitive). Since DAG minimization only compresses repeated
subtrees, it misses many internal repeats, and is thus insufficient in many
cases.

Bille et al.\ introduced tree compression with top
trees~\cite{TopTrees2013}, which this paper builds upon.
Their method exploits both repeated subtrees and tree structure repeats, and
can compress exponentially better than DAG minimization.
They give a $\log_\NumS^{0.19}{n}$
worst-case compression ratio for a tree of size $n$ labelled from an alphabet of
size $\NumS$ for their algorithm. They show that navigation and a
number of other operations are supported in $O(\log n)$ time directly
on the compressed representation.
However, they do not give any practical evaluation, and indeed state as an
open question whether top-tree compression has practical value.

\paragraph*{Tree Grammars.} %.................................................
A popular approach to exploit the redundancy of tree patterns is to represent the tree using a
formal grammar that generates the input tree, generalizing grammar compression
from strings to trees~\cite{BPLEX2004, TreeCompression2004, GrammarComprExp2006,
GrammarCompression2005, TreeRePair2013, ApproxSmallest2013}.
These can be exponentially smaller than the minimal
DAG~\cite{GrammarComprExp2006}. Since it is NP-Hard to compute the smallest
grammar~\cite{SmallestGrammar2005}, efficient heuristics are required.

One very simple yet efficient heuristic method is RePair~\cite{RePair2000}. A
string compressor, it can be applied to a parentheses bitstring representation
of the tree. The output grammars produced by RePair can support a variety
of navigational operations and random access, in time logarithmic in the input
tree size, after additional processing~\cite{RandomAccessGrammars2011}. These
methods, however, appear to require significant engineering effort
before their practicality can be assessed.

TreeRePair~\cite{TreeRePair2013} is a generalization of RePair
from strings to trees. It achieves the best grammar compression ratios currently
known. However, navigating TreeRePair's grammars in sublinear time with respect
to their depth, which can be linear in their size~\cite{TopTrees2013}, is an
open problem.  For relatively small documents (where the output of
TreeRePair fits in cache), the navigation
speed for simple tree traversals
is about 5 times slower than succinct representations~\cite{TreeRePair2013}.

Several other popular grammar compressors exist for trees. Among them,
BPLEX~\cite{BPLEX2004,GrammarCompression2005} is probably best-known,
but is much slower than TreeRePair.
The~\textsf{TtoG} algorithm is the first to achieve a good theoretical
approximation ratio~\cite{ApproxSmallest2013}, but has not been evaluated in
practice.

%...............................................................................
\paragraph*{Succinct Representations.} Another approach is to represent the
tree using near-optimal space without applying compression methods to its
structure, a technique called \textit{succinct data structures}. Unlabelled
trees can be represented using~$2n+\oh{n}$ bits~\cite{JacobsonSuccinct1989} and
support queries in constant time~\cite{MunroSuccinct2001}.
There are a few $n \log{\NumS} + \Oh{n}$ bit-representations for labelled
trees, most notably that by Ferragina et al.~\cite{FerraginaSuccinct2009},
which also yields a compressor, XBZip. While XBZip has good performance
on XML files \emph{in their entirety}, including text, attributes etc.,
evidence suggests that it does not beat TreeRePair on pure labelled trees.
As the authors admit, it is also slow.

\begin{figure}[tb]
\center
\mbox{
	\subfigure[]{
		\label{fig:mergetypes:a}
		\begin{tikzpicture}
			\draw[black,fill=black] (1,2.5) circle (3pt);
			\draw[black] (1,1.5) circle (3pt);
			\draw[black,fill=black] (1,0.5) circle (3pt);
			\draw (1,2) ellipse (8pt and 20pt);
			\draw (1,1) ellipse (8pt and 20pt);
		\end{tikzpicture}
	}
	\quad
	\subfigure[]{
		\label{fig:mergetypes:b}
		\begin{tikzpicture}
			\draw[black,fill=black] (1,2.5) circle (3pt);
			\draw[black] (1,1.5) circle (3pt);
			\draw (1,2) ellipse (8pt and 20pt);
			\draw (1,1) ellipse (8pt and 20pt);
		\end{tikzpicture}
	}
	\quad
	\subfigure[]{
		\label{fig:mergetypes:c}
		\begin{tikzpicture}
			\draw[black, fill=black] (1,1.5) circle (3pt);
			\draw[black, fill=black] (0.46,0.65) circle (3pt);
			\draw[rotate around={30:(1,1.5)}] (1.05,1) ellipse (8pt and 20pt);
			\draw[rotate around={-30:(1,1.5)}] (0.95,1) ellipse (8pt and 20pt);
		\end{tikzpicture}
	}
	\quad
	\subfigure[]{
		\label{fig:mergetypes:d}
		\begin{tikzpicture}
			\draw[black, fill=black] (1,1.5) circle (3pt);
			\draw[black, fill=black] (1.54,0.65) circle (3pt);
			\draw[rotate around={30:(1,1.5)}] (1.05,1) ellipse (8pt and 20pt);
			\draw[rotate around={-30:(1,1.5)}] (0.95,1) ellipse (8pt and 20pt);
		\end{tikzpicture}
	}
	\quad
	\subfigure[]{
		\label{fig:mergetypes:e}
		\begin{tikzpicture}
			\draw[black, fill=black] (1,1.5) circle (3pt);
			\draw[rotate around={30:(1,1.5)}] (1.05,1) ellipse (8pt and 20pt);
			\draw[rotate around={-30:(1,1.5)}] (0.95,1) ellipse (8pt and 20pt);
		\end{tikzpicture}
	}
}

\caption{Five kinds of cluster merges in top trees. Solid nodes are boundary
nodes, hollow ones are boundary nodes that become internal. Source of this
graphic and more details:~\cite[Section 2.1]{TopTrees2013}.}\label{fig:mergetypes}
\end{figure}
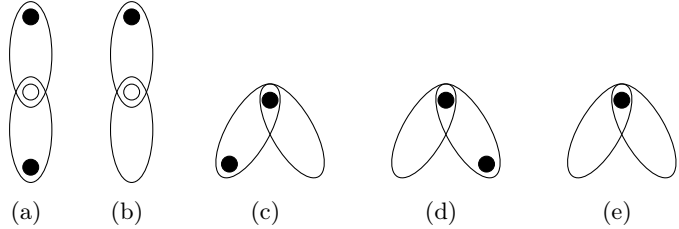

%=========================================================================
%  Preliminaries
%=========================================================================

%=========================================================================
%  Our Results
%=========================================================================
\subsection{Our Results}\label{s:intro:results}

Our primary aim in this paper is to address the question
of Bille et al.~\cite{TopTrees2013} regarding the
practicality of the top tree approach, but we make some
theoretical contributions as well.  We first give some 
terminology and notation.

A \emph{top tree}~\cite{TopTrees2005} is a hierarchical decomposition of a
tree into \emph{clusters}, which represent subgraphs of the original tree.
Leaf clusters correspond to single edges, and inner clusters represent the union
of the subgraphs represented by their two children. Clusters are formed in one
of five ways, called \emph{merge types}, shown in Figure~\ref{fig:mergetypes}. A
cluster can have one or two \textit{boundary nodes}, a top- and optionally a
bottom boundary node, where other clusters can be attached by merging. A top
tree's minimal DAG is referred to as a \textit{top DAG}. For further details on
the fundamentals of tree compression with top trees, refer to~\cite{TopTrees2013}. Throughout this paper, let~$\Tree$ be any ordered, labelled tree with~$\NumN$
nodes, and let~$\Alphabet$ denote the label alphabet 
with~$\NumS := \card{\Alphabet}$. Let~$\Toptree$ be the top tree and~$\TopDAG$ the top DAG
corresponding to~$\Tree$, and~$\NumNTD$ the total size (nodes plus edges) of~$\TopDAG$. We assume a standard word RAM model with 
logarithmic word size, and measure space
complexity in terms of the number of words used.  
Then:
\begin{theorem}\label{thm:dagsize}
The size of the top DAG is $\NumNTD = 
\Oh{\frac{\NumN}{\log_\NumS{\NumN}} \cdot \log\log_\NumS{\NumN}}$.
\end{theorem}

\noindent This is only a factor of $\Oh{\log \log _\NumS{\NumN}}$
away from the information-theoretic lower bound, and greatly 
improves the bound of $\Oh{n/\log_\NumS^{0.19}{n}}$ obtained by
Bille et al. and answers an open question posed in a talk by
Weimann.  

Next, we show that if only
basic navigation is to be performed, the amount of information that
needs to be stored can be greatly reduced, relative
to the original representation~\cite{TopTrees2013},
without affecting the asymptotic running time.
\begin{theorem}\label{thm:nav}
We can support navigation with the operations \textsf{Is
Leaf}, \textsf{Is Last Child}, \textsf{First Child}, \textsf{Next Sibling}, and
\textsf{Parent} in $\Oh{\log{\NumN}}$ time, full decompression in time
$\Oh{\NumN}$
%, and \dodgy[given node ID in $\TopDAG$...]{label access in time~$\Oh{1}$} 
on a representation of size $\Oh{\NumNTD}$ storing only the top
DAG's structure, the merge types of inner nodes (an integer from $[1..5]$), 
and leaves' labels.
\end{theorem}
We believe this approach will have
low overhead and fast running times in practice for in-memory
navigation without decompression, and sketch how one would approach
an implementation.

Furthermore, we introduce the notion of \textit{combiners} that determine the
order in which clusters are merged during top tree construction. 
Combiners aim to improve the compressibility of the top tree, resulting in a
smaller top DAG. We present one such combiner that
applies the basic idea of RePair~\cite{RePair2000} to top tree compression,
prioritizing merges that produce subtree repeats in the top tree, in
Section~\ref{s:combiner}.  We give a relatively naive encoding of the
top tree, primarily using Huffman codes, and evaluate its compression
performance.  Although the output of the modified top tree compressor
is up to 50\,\% larger than the state-of-the-art TreeRePair, it is about
six times faster.  We believe that the compression gap can be narrowed
while maintaining the speed gap.

%=========================================================================
%  Content
%=========================================================================
\section{Top Trees Revisited}\label{s:contrib}

%...............................................................................
\subsection{DAG Design Decisions}\label{s:contrib:dag}

The original top tree compression paper~\cite{TopTrees2013} did not try to
minimize the amount of information that actually needs to be stored. Instead,
the focus was on implementing a wide variety of navigation operations in
logarithmic time while maintaining~$\Oh{\NumNTD}$ space \textit{asymptotically}.
Here, we reduce the amount of additional information stored about the clusters
to obtain good compression ratios.

Instead of storing the labels of both endpoints of a leaf cluster's
corresponding edge, we store only the child's label, not the parent's. 
In addition to reducing
storage requirements, this reduces the top tree's alphabet size from~$\NumS^2+5$
to~$\NumS+5$, as each cluster has either one label or a merge type. This
increases the likelihood of identical subtrees in the top tree, improving
compression. Note that this change implies that there is exactly one leaf
cluster in the top DAG for each distinct label in the input. To code the root,
we perform a merge of type~(a) (see Section~\ref{s:intro:results} and
Figure~\ref{fig:mergetypes}) between a dummy edge leading to the root and the
last remaining edge after all other merges have completed.

With these modifications, we reduce the amount of information stored with the
clusters to the bare minimum required for decompression, i.e.\ leaf clusters'
labels and inner clusters' merge types.

Lastly, we speed up compression by directly constructing the top DAG
during the merge process. We initialize it with all distinct leaves, and
maintain a mapping from cluster IDs to node IDs in~$\TopDAG$, as well as a hash
map mapping DAG nodes to their node IDs. When two edges are merged into a new
cluster, we look up its children in the DAG and only need to add a new node
to~$\TopDAG$ if this is its first occurrence. Otherwise, we simply update the
cluster-to-node mapping.

%...............................................................................
\subsection{Navigation}\label{s:contrib:nav}

We now explain how to navigate the top DAG with our reduced information set. We
support full decompression in time $\Oh{\NumN}$, as well as operations to move
around the tree in time proportional to the height of the top DAG,
i.e.~$\Oh{\log{\NumN}}$. These are: determining whether the current
node is a leaf or its parent's last child, and moving to its first child, next
sibling, and parent. Accessing a node's label is possible in constant time given
its node number in the top DAG.

\par\bigskip
\begin{proof}[ (Theorem~\ref{thm:nav})]
As a node in a DAG can be the child of any number of other nodes, it does not
have a unique parent. Thus, to allow us to move back to a node's parent in the DAG, we need to maintain
a stack of parent cluster IDs along with a bit to indicate whether we descended into
the left or right child. We refer to this as the \textit{DAG stack}, and update
it whenever we move around in~$\TopDAG$ with the operations below. Similarly,
we also maintain a \textit{tree stack} containing the DAG stack of each ancestor
of the current node in the (original) tree.

\paragraph*{Decompression:} We traverse the top DAG in pre-order, undoing the merge
operations to reconstruct the tree. We begin with~$\NumN$ isolated nodes,
and then add back the edges and labels as we traverse the top DAG. As this
requires constant time per cluster and edge, we can decompress the top DAG
in~$\Oh{\NumN}$~time.

\paragraph*{Label Access:} Since only leaf clusters store labels, and these are
coded as the very first clusters in the top DAG (cf. Section~\ref{s:contrib:%
encoding}), their node indices come before all other nodes'. Therefore,
a leaf's label index~$i$ is its node number in the top DAG. We access the label
array in the position following the~$(i-1)$th null byte, which we can find with
a~$\mathsf{Select}_{0}(i-1)$ operation, and decode the label string until we
reach another null byte or the end.

\paragraph*{Is Leaf:} A node is a leaf iff it is no cluster's top boundary node.
Moving up through the DAG stack, if we reach a cluster of
type~(a) or~(b) from the \textit{left} child, the node is not a leaf (the
left child of such a cluster is the \textit{upper} one in
Figure~\ref{fig:mergetypes}). If, at any point before encountering such a
cluster, we exhaust the DAG stack or reach a cluster of type~(b) or~(c) from the right, type~(d) from the
left, or type~(e) from either side, the node is a leaf.
This can again be seen in Figure~\ref{fig:mergetypes}.

\paragraph*{Is Last Child:} We move up the DAG stack until we reach a cluster of
type~(c),~(d), or~(e) from its left child. Upon encountering a cluster of
type~(a) or~(b) from the right, or emptying the DAG stack completely, we abort
as the upward search lead us to the node's parent or exhausted the tree, respectively.

\paragraph*{First Child and Next Sibling:} First, we check whether the node is a
leaf (\textsf{First Child}) or its parent's last child (\textsf{Next Sibling}),
and abort if it is. \textsf{First Child} then pushes a copy of the DAG stack onto the tree
stack. Next, we re-use the upward search performed by the previous check,
removing the elements visited by the search from the DAG stack, up until the
cluster with which the search ended. We descend into its right child and keep
following the left child until we reach a leaf.

\paragraph*{Parent:} Since \textsf{First Child} pushes the DAG stack
onto the tree stack, we simply reset the DAG stack to the tree stack's top
element, which is removed.
\end{proof}

We note here that the tree stack could, in theory, grow to a size of~$\Oh{\NumN
\log{\NumN}}$, as the tree can have linear height and the logarithmically sized
DAG stack is pushed onto it in each \textsf{First Child} operation. However, we
argue that due to the low depth of common labelled trees, especially XML files,
this stack will remain small in practice. Even when pessimistically assuming a
\emph{very} large tree with a height of~80 nodes, with a top tree of height~50,
the tree stack will comfortably fit into~32\,kB when using~64-bit node IDs. Our
preliminary experiments confirm this.

To improve the worst-case tree stack size in theory, we can instead keep a log
of movements in the top DAG, which is limited in size to the distance travelled
therein. We expect this to be significantly less than~$\Oh{\NumN \log{\NumN}}$
in expectation.

%...............................................................................
\subsection{Worst-Case Top DAG size}
\label{contrib:bounds}

Bille et al. show that a tree's top tree has at most~$\Oh{\NumN / \log_{\NumS}
^{0.19}{\NumN}}$ distinct clusters~\cite{TopTrees2013}. This bound, however, is
an artifact of the proof. By modifying the definition of a
\textit{small cluster} in the compression analysis and carefully exploiting the
properties of top trees, we are able to show a new, tighter, bound, which
directly translates to an improvement on the worst-case compression ratio.
Before we can prove Theorem~\ref{thm:dagsize}, we need to show the following
essential lemmata. Let~$\SubS{\vertex}$ be the size of~$\vertex$'s subtree,
and~$p(\vertex)$ denote its parent.

\begin{lemma}\label{lem:logheight}
Let~$\Tree$ be any ordered labelled tree of size~$\NumN$, and let~$\Toptree$ be
its top tree. For any node~$\vertex$ of~$\Toptree$, the height of its subtree is
at most~$\lfloor\log_{8/7}{\SubS{\vertex}}\rfloor$.
\end{lemma}
\begin{proof}
Consider the incremental construction process of a top tree~$\Toptree$. During
the merge process, the algorithm builds up a tree by joining clusters into
larger clusters. We start with a forest of~$\NumN+1$ nodes, each representing an
edge of~$\Tree$. Every merge operation joins two clusters and thus reduces the number
of connected components in $\Toptree$ by one. These connected components are
subtrees of the final top tree. We can thus think of them as the top trees for
tree patterns of the input tree.

Note that a subtree of the top tree is not the top tree of a rooted subtree for
two reasons. For one, it might represent some, but not all, siblings of a node.
This is due to horizontal merges operating on pairs of edges to consecutive
siblings. Thus, a cluster could, for example, represent a node and the subtrees
of the first two of its five children. Secondly, if the cluster has a bottom
boundary node~$\vertexC$ (drawn as a filled node at the bottom in Figure~\ref{fig:mergetypes}),
the subtree of~$\Tree$ that is rooted at~$\vertexC$ is \textit{not} contained in the
cluster. Thus, the cluster does not correspond to a subtree of~$\Tree$, but
rather a tree pattern, i.e.\ a connected subgraph.

Therefore, a subtree of a top tree is the top tree of a tree pattern of~$\Tree$,
and the same bounds apply to its height. As each iteration of merges in the top
tree construction reduces the number of strongly connected components by a factor of~$c \geq
8/7$~\cite{TopTrees2013}, there are at most~$\lceil\log_{8/7}{\NumN}\rceil$
iterations, each of which increases the height of the top tree by exactly 1.
Being a full binary tree with~$\NumN+\nobreak1$ leaves, the top tree has~$2 \NumN$
edges. Thus, the height of any top tree of size~$n$ is bounded
by~$\lceil\log_{8/7}{\frac{n}{2}}\rceil <\nobreak \lfloor \log_{8/7}{n} \rfloor \approx\nobreak 5.2\log{n}$.
By the above, this also applies to subtrees of top trees. \end{proof}

\begin{lemma}\label{lem:smallclusters}%
Let~$\Tree$ be any ordered labelled tree of size~$\NumN$, let~$\Toptree$ be its
top tree, and~$\thresh$ be an integer. Then~$\Toptree$ contains at
most~$\Oh{(\NumN / \thresh) \cdot \log{\thresh}}$ nodes~$\vertex$ so
that~$\SubS{\vertex} \leq \thresh$ and~$\SubS{p(\vertex)} > \thresh$.
\end{lemma}
\begin{proof}
We will call any node~$\vertex$ of the top tree a \textit{light} node
iff~$\SubS{\vertex} \leq \thresh$, otherwise we refer to it as \textit{heavy}.
With this terminology, we are looking to bound the number of light nodes whose
parent is heavy.

As~$\Toptree$ is a full binary tree, there are four cases to distinguish. We are
not interested in the children of light nodes, nor are we interested in heavy
nodes with two heavy children. This leaves us with two interesting cases:
\begin{enumerate}
	\item A heavy node~$\vertexA$ with two light children~$\vertexB$
	and~$\vertexC$. Then,~$\SubS{\vertexB} + \SubS{\vertexC} \geq \thresh$.
	Thus, there are at most~$2\NumN/\thresh = \Oh{\NumN / \thresh}$ light nodes
	with a heavy parent and a light sibling.

	\item A heavy node with one light and one heavy child. We will consider
	this case in the remainder of the proof.
\end{enumerate}

Consider any heavy node~$\vertex$. We say that~$\vertex$ is in
\textit{class~$i$} iff~$\SubS{\vertex} \in \left[2^i,2^{i+1}\!-1\right]$. Observe
that only classes~$i \geq \lfloor\log_2{\thresh}\rfloor$ can contain heavy nodes,
and that the highest non-empty class is~$\lfloor\log_2{\NumN}\rfloor$. Let a
\textit{top class~$i$ node} be a node in class~$i$ whose parent is in class~$j >
i$, and a \textit{bottom class~$i$ node} one for which both children are in
classes lower than~$i$. We now make two propositions:

\begin{description}
	\item[Proposition 1] A node~$\vertexA$ of class~$i$ can have at most one
	child in class~$i$.%
  \\\emph{Proof:} Assume both children~$\vertexB, \vertexC$ of $\vertexA$ are in
  class~$i$. Then, the subtree size of~$\vertexA$ is~$\SubS{\vertexA} = 1 +
  \SubS{\vertexB} + \SubS{\vertexC} \geq 1 + 2^i + 2^i > 2^{i+1}$, and thus by
  definition~$\vertexA$ is not in class~$i$.

	\item[Proposition 2] Let~$\vertexB$ be a top class~$i$ node. There are at
	most~$\Oh{i}$ light nodes in the subtree of~$\vertexB$ that are children of
	class~$i$ nodes.%
	\\\textit{Proof:} By Proposition 1, there exists exactly one path of
	class~$i$ nodes in the subtree of~$\vertexB$. This path begins at~$\vertexB$
	and ends at the bottom class~$i$ node of the subtree of~$\vertexB$, which we
	refer to as~$\vertexC$. There are no other
	class~$i$ nodes in the subtree of~$\vertexB$. Being a full binary tree, the
	height of~$\vertexC$'s subtree fulfills~$\Height{\vertexC} \geq\nobreak
	\log_2{\SubS{\vertexC}} \geq\nobreak \log_2{2^i} = i$. We now use
	Lemma~\ref{lem:logheight} to obtain an upper bound on~$\Height{\vertexB}$ of $\Height{\vertexB} \leq
	\lfloor\log_{8/7}{\SubS{\vertexB}}\rfloor \leq \frac{i+1}{\log_2{8/7}} \approx 5.2 \cdot
	(i+1)$. Thus, the path from the top class~$i$ node to the bottom class~$i$
	node has a length of~$l \leq \Height{\vertexB} - \Height{\vertexC} =
	\Oh{i}$. Each node on the path can have at most one light child by Proposition~1. Thus, there
	are at most~$\Oh{i}$ light nodes that are children of class~$i$ nodes in the
	subtree of a top class~$i$ node.
\end{description}

Combining Proposition 2 with the observation that the number of top class~$i$
nodes is clearly at most~$\NumN / 2^i$, we obtain a bound on the number of
class~$i$ nodes with one heavy and one light child of~$\NumN / 2^i \cdot
\Oh{i}$. We then sum over all classes containing heavy nodes to obtain the total
number of heavy nodes with one light child,

\begin{equation*}
	\sum_{i=\lfloor\log_2{\thresh}\rfloor}^{\lfloor\log_2{\NumN}\rfloor}{\frac{
	\NumN}{2^i}\cdot\Oh{i}} = \Oh{\frac{\NumN}{\thresh} \cdot \log{\thresh}}
\end{equation*}

Thus, there are at most~$\Oh{\NumN / \thresh \cdot \log{\thresh}} + 2\NumN /
\thresh = \Oh{\NumN / \thresh \cdot \log{\thresh}}$ light nodes whose parent is
heavy. This concludes the proof.

\end{proof}

\begin{proof}[ (Theorem~\ref{thm:dagsize})]
We define a \textit{small cluster} as one whose subtree contains at
most~$2^\J+1$ nodes and set~$\J=\log_2{(0.5 \log_{4\NumS}{\NumN})}$. We call a
small cluster \textit{maximal} if its parent's subtree exceeds the size limit of
a small cluster. A cluster that is not small is called a \textit{large} cluster.
Note that this is a special case of our distinction between light and heavy
nodes in the proof of Lemma~\ref{lem:smallclusters}.

As each of the~$\NumN$ leaves of the top tree is contained in exactly one
maximal small cluster, and the top tree is a full binary tree, there is exactly one
large cluster less than there are maximal small clusters. Thus, it suffices to
show that there are at most~$\Oh{(\NumN \cdot \log{\log_\NumS{\NumN}) /
\log_\NumS{\NumN}}}$ maximal small clusters, and that the total number of
distinct small clusters does not exceed said bound.

Recall that each inner node of the top tree is labelled with one of the five merge types, and
that each leaf stores the label of its edge's child node, as described in
Section~\ref{s:contrib:dag}. Therefore, the top tree is labelled with an
alphabet of size~$\NumS+5 = \Oh{\NumS}$.

To bound the total number of distinct small clusters, we consider the number of
distinct labelled trees of size at most~$\SB$, which is~$\Ohsmall{(4\NumS)^{\SB+1}}$, and
can be rewritten as~$\Ohsmall{\NumS^2\sqrt{\NumN}}$ by setting~$\SB =\nobreak 2^\J+1$~\cite{TopTrees2013}.
If~$\NumS<\NumN^{1/8}$, this further reduces to~$\Ohsmall{\NumN^{3/4}}$. Otherwise, the
theorem holds trivially as $\log_\NumS{\NumN}=\Oh{1}$.

We now bound the number of maximal small clusters with
Lemma~\ref{lem:smallclusters} by choosing the threshold $\thresh =\nobreak 2^\J+\nobreak1 =\nobreak 0.5
\log_{4\NumS}{\NumN} + 1 = \Oh{\log_{\NumS}{\NumN}}$. As a maximal small cluster
is a light node whose parent is heavy, we can use Lemma~\ref{lem:smallclusters} to
bound the number of maximal small clusters by~$\Oh{(\NumN \cdot \log{\thresh})
/ \thresh} = \Oh{(\NumN \cdot \log{\log_\NumS{\NumN}}) / \log_\NumS{\NumN}}$.
This concludes the proof.
\end{proof}

% Our claim that the worst-case compression ratio is $\Oh{\log_\NumS{\NumN} /
% \log{\log_\NumS{\NumN}}}$ is a direct consequence of Theorem~\ref{thm:dagsize}.
% The previous result by Bille et al. showed a compression ratio of at
% least~$\Oh{\log_{\NumS}^{0.19}{\NumN}}$~\cite{TopTrees2013}, upon which our
% result is a significant improvement.

%...............................................................................
\subsection{Encoding}\label{s:contrib:encoding}
%We can encode the DAG as follows, which allows us to apply the navigation
%operations from Theorem~\ref{thm:nav} directly to the compressed
%representation.

\newcommand{\CT}{core\xspace}

In the top DAG, we need to be able to access a cluster's left and right child,
as well as its merge type for inner clusters or the child node's label for
the edge that it refers to for leaf clusters.
To realize this interface, we decompose the top DAG into a binary
\textit{\CT} tree and a pointer array. The \CT tree is defined by removing all
incoming edges from each node, except for the one coming from the node with
lowest pre-order number. All other occurrences are replaced by a dummy leaf node
storing the pre-order number of the referenced node. Leaves in the top DAG are
assigned new numbers as label pointers, which are smaller than the IDs of all inner nodes. All references to
leaves, including the dummy nodes, are coded in an array of \textit{pointers},
ordered by the pre-order number of the originating node. Similarly, the inner
nodes' merge types are stored in an array in pre-order. Lastly, the \CT tree itself
can be encoded using two bits per inner node, indicating whether the left and
right children are inner nodes in the \CT tree.

Using this representation, all that is required for efficient navigation is an
entropy coder providing constant-time random access to node pointers and merge
types, and a data structure providing \textsf{rank} and \textsf{select} for the
\CT tree and label strings. All of these building blocks can be treated as black
boxes, and are well-studied and readily available, e.g.~\cite{Succinter2008} and
the excellent SDSL~\cite{SDSL2014} library.

\paragraph{Simple Encoding} To obtain file size results with reasonable effort,
we now describe a very simple encoding that does not lend itself to navigation
as easily. We compress the \CT tree bitstring and merge types using blocked
Huffman coding. The pointer array and null byte-separated concatenated label
string are encoded using a Huffman code. The Huffman trees are coded like the
\CT tree above. The symbols are encoded using a fixed length and
concatenated. Lastly, we store the sizes of the four Huffman code segments as a
file header.

%=========================================================================
%  Heuristic Combiners
%=========================================================================
\section{Heuristic Combiners}\label{s:combiner}

As described in the original paper~\cite{TopTrees2013}, the construction of the
top tree \textit{exposes} internal repetitions. However, it does not attempt to
maximize the size or number of identical subtrees in the top tree, i.e.\ its
compressibility. Instead, the merge process sweeps through the tree linearly
from left to right and bottom to top. This is a straight-forward cluster
combining strategy that fulfills all the requirements for constructing a top
tree, but does not attempt to maximize compression performance. We therefore
replace the standard combining strategy with heuristic methods that try to increase
compressibility of the top tree. Here, we present one such combiner that applies
the basic idea of RePair to the horizontal merge step of top tree compression.
(In preliminary experiments, it proved detrimental to apply the heuristic to
vertical merges, and we limit ourselves to the horizontal merge step, but note
that this is not a general restriction on combiners.)

We hash all clusters in the top tree as they are created. The hash value
combines the cluster's label, merge type, and the hashes of its left and right
children if these exist. As the edges in the auxiliary tree correspond to
clusters in the top tree during its construction, we assign the cluster's hashes
to the corresponding edges. Defining a digram as two edges whose clusters can
be merged with one of the five merge types from Figure~\ref{fig:mergetypes}, we
can apply the idea of RePair, identifying the edges by their hash values. In
descending order of digram frequency, we merge all non-overlapping occurrences,
updating the remaining edges' hash values to those of the newly created
clusters.

Since this procedure does not necessarily merge a constant fraction of the edges
in each iteration, we may need to additionally apply the normal horizontal merge
algorithm if too few edges were merged by the heuristic. The constant upon which
this decision is based thus becomes a tuning parameter. Note that we need to
ensure that every edge is merged at most once per iteration.

%=========================================================================
%  Evaluation
%=========================================================================
\section{Evaluation}\label{s:eval}

% IMPORT-DATA input eval/input.txt
\begin{table}[t]
\caption{XML corpus used for our experiments. File sizes are given for stripped
documents, i.e.\ after removing whitespace and tags' attributes and contents.}\label{tbl:corpus}
\footnotesize
\centering
\begin{minipage}{.497\linewidth}
\begin{tabular}{l rrr}
\toprule
\textbf{File name} & \textbf{size (MB)} & \textbf{~~~\#\,nodes} & \textbf{height}\\\midrule
%% TABULAR REFORMAT(col 1=(precision=2 group=\,) col 2=(precision=0 group=\,))
%% SELECT SUBSTR(REPLACE(SUBSTR(SUBSTR(file, INSTR(file, "/")+1),-4,-99),"_","-"), 1, 15) as x,
%% (strippedSize/1000000.0) AS y, nodes as z, height FROM input
%% WHERE x != "NCBI-snp.chr1" AND x != "factor0.2" AND x != "factor1" AND x != "factor2" AND x != "dblp-small" AND x != "enwiki-new"
%% ORDER BY lower(x) LIMIT 8
 1998statistics &   0.60 &      28\,306 &  6 \\
           dblp & 338.87 & 20\,925\,865 &  6 \\
enwiki-latest-p & 229.78 & 14\,018\,880 &  5 \\
       factor12 & 359.36 & 20\,047\,329 & 12 \\
        factor4 & 119.88 &  6\,688\,651 & 12 \\
      factor4.8 & 143.80 &  8\,023\,477 & 12 \\
        factor7 & 209.68 & 11\,697\,881 & 12 \\
  JST-gene.chr1 &   5.79 &     173\,529 &  7 \\
% END TABULAR SELECT SUBSTR(REPLACE(SUBSTR(SUBSTR(file, INSTR(file, "/")+1),-...
\bottomrule
\end{tabular}
\end{minipage}
\begin{minipage}{.497\linewidth}
\begin{tabular}{l rrr}
\toprule
\textbf{File name} & \!\textbf{size (MB)} & \textbf{~~~\#\,nodes} & \textbf{height}\\\midrule
%% TABULAR REFORMAT(col 1=(precision=2 group=\,) col 2=(precision=0 group=\,))
%% SELECT SUBSTR(REPLACE(SUBSTR(SUBSTR(file, INSTR(file, "/")+1),-4,-99),"_","-"), 1, 15) as x,
%% (strippedSize/1000000.0) AS y, nodes as z, height FROM input
%% WHERE x != "NCBI-snp.chr1" AND x != "factor0.2" AND x != "factor1" AND x != "factor2" AND x != "dblp-small" AND x != "enwiki-new"
%% ORDER BY lower(x) LIMIT 8,8
  JST-snp.chr1 &  27.31 &     803\,596 &  8 \\
          nasa &   8.43 &     476\,646 &  8 \\
NCBI-gene.chr1 &  35.30 &  1\,065\,787 &  7 \\
      proteins & 365.12 & 21\,305\,818 &  7 \\
     SwissProt &  45.25 &  2\,977\,031 &  5 \\
    treebank-e &  25.92 &  2\,437\,666 & 36 \\
           uwm &   1.30 &      66\,729 &  5 \\
          wiki &  42.29 &  2\,679\,553 &  5 \\
% END TABULAR SELECT SUBSTR(REPLACE(SUBSTR(SUBSTR(file, INSTR(file, "/")+1),-...
\bottomrule
\end{tabular}
\end{minipage}
\end{table}

We now present an experimental evaluation of top tree compression. In this
section, we demonstrate its qualities as a fast and efficient compressor,
compare it against other compressors, and show the effectiveness of our
RePair-inspired combiner.

\paragraph*{Experimental Setup} %...............................................
All algorithms were implemented in C++11 and compiled with the GNU C++ compiler
\texttt{g++} in version 4.9.2 using optimization level~\texttt{fast} and
profile-guided optimizations. The experiments were conducted on a commodity PC
with an Intel Core i7-4790T CPU and 16\,GB of DDR3 RAM, running Debian Linux
from the \texttt{sid} channel. We used \texttt{gzip 1.6-4} and \texttt{bzip2
1.0.6-7} from the standard package repositories. Default compression settings
were used for all compressors, except the \texttt{-9} flag for gzip. All input
and output files were located in main memory using a \texttt{tmpfs} RAM disk to
eliminate I/O delays.

\paragraph*{XML corpus} %.......................................................
We evaluated the compressor and our heuristic improvement on a corpus of common
XML files~\cite{UWXML,JSNP,WikiXML,Delpratt2009}, listed in
Table~\ref{tbl:corpus}. In our experiments, we give file sizes for our simple
encoding, which represent pessimistic results that can serve as an upper bound
of what to expect from a more optimized encoding. We give these file sizes to
demonstrate that even a simple encoding yields good results with regard to
file size, speed, and ease of navigation (see Section~\ref{s:contrib:nav}).

% IMPORT-DATA eval eval/result.txt
% IMPORT-DATA repair eval/repair.txt
% IMPORT-DATA treerepair eval/treerepair.txt

\begin{figure}[t]
\includegraphics{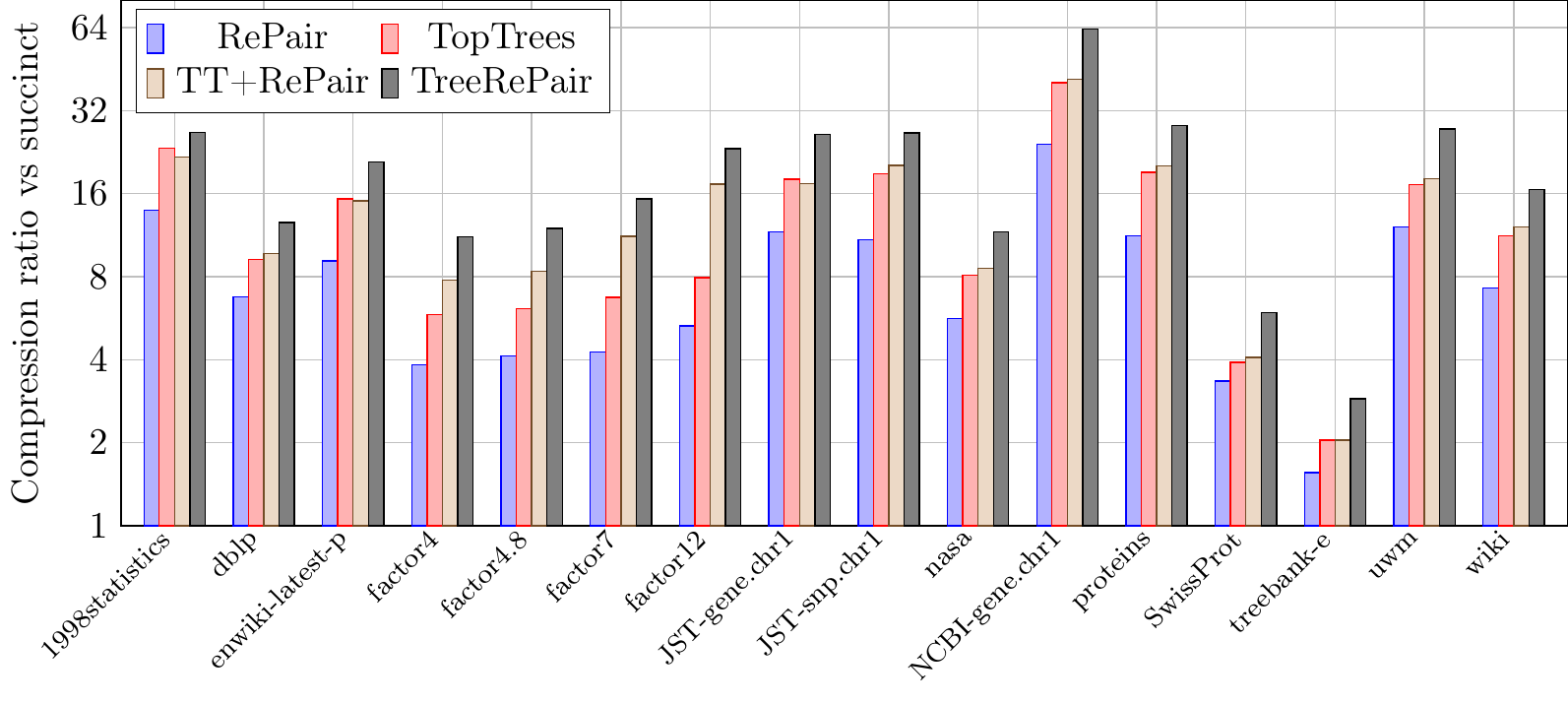}
\caption{Comparison of compression ratios, measured by comparing file sizes
against a succinct encoding of the input file (higher is better)}
\label{fig:vs_succ}
\end{figure}

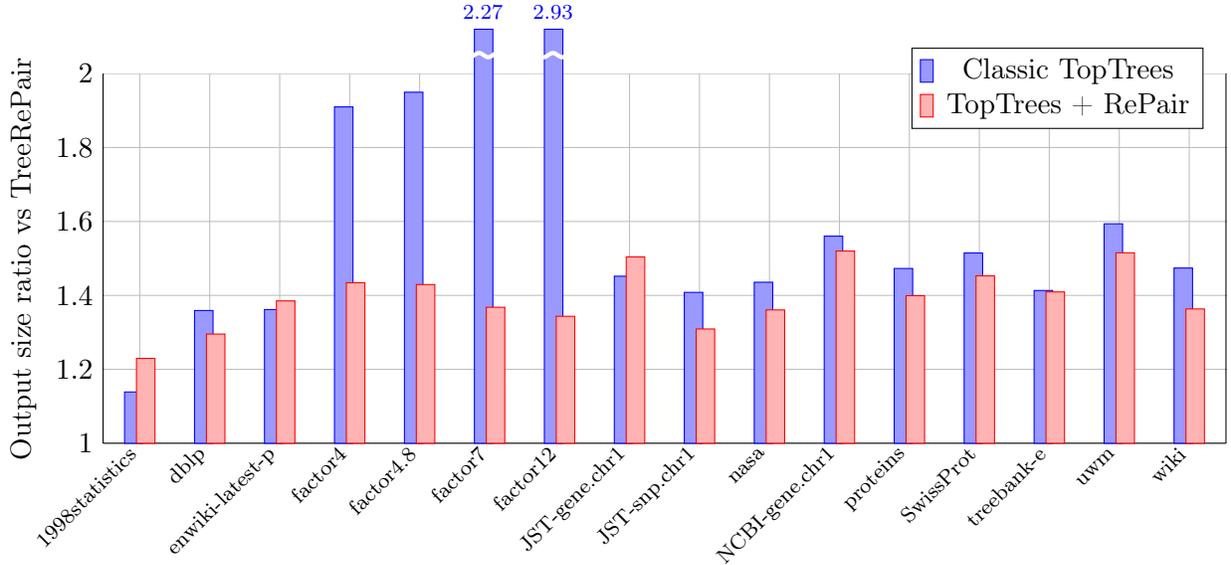
\begin{figure}[bt]
\begin{tikzpicture}
  \begin{axis}[
    axis x line=box,
    axis y line=box,
    axis x line*=none,
    height=6.5cm,
    width=\columnwidth,
    grid=major,
    ylabel={Output size ratio vs TreeRePair},
    % x ticks
    symbolic x coords={1998statistics,dblp,enwiki-latest-p,factor4,factor4.8,factor7,factor12,JST-gene.chr1,JST-snp.chr1,nasa,NCBI-gene.chr1,proteins,SwissProt,treebank-e,uwm,wiki},
    xtick=data,
    x tick label style={rotate=45, anchor=east, font=\fontsize{8pt}{1em}\selectfont},
    major tick length=0pt, %hide ticks
    ymin=1,
    ymax=2,
    ybar=-2.5pt, % overlapping bars
    restrict y to domain*=0:2.12,
    clip=false,
    visualization depends on=rawy\as\rawy, % Save the unclipped values
    after end axis/.code={ % Draw line indicating break
        \draw [ultra thick, white, decoration={snake, amplitude=0.8pt, segment length=7pt}, decorate] (rel axis cs:0,1.05) -- (rel axis cs:1,1.05);
    },
    % print nodes only for clipped values
    every node near coord/.style={
        check for zero/.code={
	        \pgfkeys{/pgf/fpu=true}
	        \pgfmathparse{\pgfplotspointmeta-2}f
	        \pgfmathfloatifflags{\pgfmathresult}{-}{
	            \pgfkeys{/tikz/coordinate}
            }{}
            \pgfkeys{/pgf/fpu=false}
        }, check for zero, xshift=-2.5, font=\scriptsize
    },
    nodes near coords={
        \pgfmathprintnumber{\rawy}% Print unclipped values
    },
    bar width=7pt,
    % x axis range
    xmin=1998statistics,
    xmax=wiki,
    enlarge x limits=0.035,
    % legend
    legend style={column sep=1mm, at={(0.98,1.07)}, anchor=north east},
    legend image code/.code={\draw[#1,draw=#1] (0cm,-0.15cm) rectangle (0.17cm,0.15cm);}
  ]

  %% PLOT SELECT SUBSTR(REPLACE(SUBSTR(SUBSTR(eval.file, INSTR(eval.file, "/")+1),-4,-99),"_","-"), 1, 15) as x,
  %% compressed*1.0/(8*treerepair.size_bytes) as y
  %% FROM eval JOIN treerepair ON eval.file = treerepair.file AND eval.repair=0
  %% WHERE x != "factor2"
  %% ORDER BY x
  \addplot[blue, fill=blue!40] coordinates { (1998statistics,1.13873) (JST-gene.chr1,1.4519) (JST-snp.chr1,1.40815) (NCBI-gene.chr1,1.56029) (SwissProt,1.51475) (dblp,1.35909) (enwiki-latest-p,1.36169) (factor12,2.93229) (factor4,1.90991) (factor4.8,1.9496) (factor7,2.27386) (nasa,1.43556) (proteins,1.4728) (treebank-e,1.41326) (uwm,1.59361) (wiki,1.47404) };
  \addlegendentry{Classic TopTrees}

  %% PLOT SELECT SUBSTR(REPLACE(SUBSTR(SUBSTR(eval.file, INSTR(eval.file, "/")+1),-4,-99),"_","-"), 1, 15) as x,
  %% compressed*1.0/(8*treerepair.size_bytes) as y
  %% FROM eval JOIN treerepair ON eval.file = treerepair.file AND eval.repair=1 AND eval.minratio=1.26
  %% WHERE x != "factor2"
  %% ORDER BY x
  \addplot coordinates { (1998statistics,1.22959) (JST-gene.chr1,1.50405) (JST-snp.chr1,1.30914) (NCBI-gene.chr1,1.51995) (SwissProt,1.45281) (dblp,1.29538) (enwiki-latest-p,1.38513) (factor12,1.34351) (factor4,1.43418) (factor4.8,1.42938) (factor7,1.36789) (nasa,1.3607) (proteins,1.39921) (treebank-e,1.40945) (uwm,1.51519) (wiki,1.3633) };
  \addlegendentry{TopTrees + RePair}

  \end{axis}
\end{tikzpicture}
\caption{Comparison of output file sizes produced by top tree compression with
and without the RePair combiner, measured against TreeRePair file sizes (lower
is better)}\label{fig:vs_trp}
\end{figure}

\paragraph*{Results} %..........................................................
We use a minimum merge ratio of $c=1.26$ for the horizontal merge step using our
RePair-inspired heuristic combiner in all our experiments. This is the result of an
extensive evaluation which showed that values $c \in [1.2, 1.27]$ work
very well on a broad range of XML documents. We observed that values close to 1
can improve compression by up to 10\,\% on some files, while causing a
deterioration by a similar proportion on others. Thus, while better choices
of~$c$ exist for individual files, we chose a fixed value for all files to
provide a fair comparison, similar to the choice of 4 as the maximum rank of the
grammar in TreeRePair~\cite{TreeRePair2013}.

We use a parenthesis bitstring encoding of the input tree as a baseline to
measure compression ratios. The unique label strings are concatenated, separated
by null bytes. Indices into this array are stored as fixed-length numbers
of~$\lceil \log_2{\#\text{labels}} \rceil$ bits.
TreeRePair\footnote{\url{https://code.google.com/p/treerepair}\label{fn:trp}},
which has been carefully optimized to produce very small output files, serves us
as a benchmark. We are, however, reluctant to compare tree compression with top
trees to TreeRePair directly, as our methods have not been optimized to the same
degree.

In Figure~\ref{fig:vs_succ} we give a compression ratios relative to the
succinct encoding. We evaluated our implementation of top tree compression using
the combining strategy from~\cite{TopTrees2013} as well as our RePair-inspired
combiner. We also give the file sizes achieved by TreeRePair and those of RePair
on a parentheses bitstring representation of the input tree and the concatenated
nullbyte-separated label string (note that no deduplication is performed here,
as this is up to the compressor). We represent RePair's grammar production rules
as a sequence of integers with an implicit left-hand side and encode this
representation using a Huffman code. Figure~\ref{fig:vs_succ} shows that top
tree compression consistently outperforms RePair already, but does not achieve
the same level of compression as TreeRePair at this stage. We can also clearly
see the impact of our RePair-inspired heuristic combiner, which improves
compression on nearly all files in our corpus and is studied in more detail in
the next paragraph. Table~\ref{tbl:size} gives the exact numbers for the output
file sizes, supplementing them with results for general-purpose compressors.

%% DEFMACRO REFORMAT(col 0,1=(precision=2) col 2=(precision=0))
%% SELECT MAX(compressed*1.0/(8*treerepair.size_bytes)) as discrepancyFactorClassic,
%%   MAX(compressed*1.0/(8*treerepair.size_bytes))-1 as discrepancyFactorClassicRel,
%%   100*(MAX(compressed*1.0/(8*treerepair.size_bytes))-1) as discrepancyFactorClassicPercent,
%%	 REPLACE(SUBSTR(SUBSTR(eval.file, INSTR(eval.file, "/")+1),-4,-99),"_","-") as discrepancyFileClassic
%% FROM eval JOIN treerepair ON eval.file = treerepair.file AND eval.repair=0
%% WHERE eval.file != "xml/factor2.xml"
\def\discrepancyFactorClassic{2.93}
\def\discrepancyFactorClassicRel{1.93}
\def\discrepancyFactorClassicPercent{193}
\def\discrepancyFileClassic{factor12}

%% DEFMACRO REFORMAT(col 0,1=(precision=2) col 2=(precision=0))
%% SELECT MAX(compressed*1.0/(8*treerepair.size_bytes)) as discrepancyFactorRepair,
%%    MAX(compressed*1.0/(8*treerepair.size_bytes))-1 as discrepancyFactorRepairRel,
%% 		100*(MAX(compressed*1.0/(8*treerepair.size_bytes))-1) as discrepancyFactorRepairPercent,
%%	 REPLACE(SUBSTR(SUBSTR(eval.file, INSTR(eval.file, "/")+1),-4,-99),"_","-") as discrepancyFileRepair
%% FROM eval JOIN treerepair ON eval.file = treerepair.file AND eval.repair=1 AND eval.minratio=1.26
%% WHERE eval.file != "xml/factor2.xml"
\def\discrepancyFactorRepair{1.52}
\def\discrepancyFactorRepairRel{0.52}
\def\discrepancyFactorRepairPercent{52}
\def\discrepancyFileRepair{NCBI-gene.chr1}

%% DEFMACRO REFORMAT(precision=2)
%% SELECT AVG(compressed*1.0/(8*treerepair.size_bytes)) as avgDiscrepancyFactorClassic,
%%    MEDIAN(compressed*1.0/(8*treerepair.size_bytes)) as medDiscrepancyFactorClassic
%% FROM eval JOIN treerepair ON eval.file = treerepair.file AND eval.repair=0
%% WHERE eval.file != "xml/factor2.xml"
\def\avgDiscrepancyFactorClassic{1.64}
\def\medDiscrepancyFactorClassic{1.47}

%% DEFMACRO REFORMAT(precision=2)
%% SELECT AVG(compressed*1.0/(8*treerepair.size_bytes)) as avgDiscrepancyFactorRepair,
%%    MEDIAN(compressed*1.0/(8*treerepair.size_bytes)) as medDiscrepancyFactorRepair
%% FROM eval JOIN treerepair ON eval.file = treerepair.file AND eval.repair=1 AND eval.minratio=1.26
%% WHERE eval.file != "xml/factor2.xml"
\def\avgDiscrepancyFactorRepair{1.39}
\def\medDiscrepancyFactorRepair{1.39}

%% DEFMACRO REFORMAT(precision=1)
%% SELECT AVG(1-repair.compressed*1.0/classic.compressed)*100 as repairImprovementAvg,
%%     MEDIAN(1-repair.compressed*1.0/classic.compressed)*100 as repairImprovementMed
%% FROM eval as repair JOIN eval AS classic ON repair.file = classic.file AND repair.repair=1 AND repair.minratio=1.26 AND classic.repair=0
%% WHERE repair.file != "xml/factor2.xml"
\def\repairImprovementAvg{10.9}
\def\repairImprovementMed{5.0}

\paragraph*{RePair Combiner.} %.................................................
\begin{absolutelynopagebreak}
Figure~\ref{fig:vs_trp} compares the two versions of top tree compression, using
TreeRePair as a benchmark. The RePair combiner's effect is clearly visible,
reducing the maximum disparity in compression relative to TreeRePair from a
file~\discrepancyFactorClassic\xspace times the size (\texttt{%
\discrepancyFileClassic}) to one that is~\discrepancyFactorRepairPercent\,\%
larger (\texttt{\discrepancyFileRepair}). This constitutes nearly a four-fold
decrease in overhead (from \discrepancyFactorClassicRel\xspace to \discrepancyFactorRepairRel).
On average, files are~\avgDiscrepancyFactorRepair\xspace times the size of TreeRePair's, down
from a factor of~\avgDiscrepancyFactorClassic\xspace before. On
our corpus, using the heuristic combiner reduced file sizes
by~\repairImprovementAvg\,\% on average, with the median being
a~\repairImprovementMed\,\% improvement compared to classical top tree
compression. Reduced compression performance was observed on few files only,
particularly smaller ones, while larger files tended to fare better.
\end{absolutelynopagebreak}

% IMPORT-DATA time eval/times_10.log eval/times_10_ttc.log eval/times_10_ttr.log

%% DEFMACRO REFORMAT(precision=1)
%% SELECT AVG(times.time / other.time) as trpVsTTclassicAvg,
%%     MEDIAN(times.time / other.time) as trpVsTTclassicMed
%% FROM time times JOIN time other ON times.file = other.file
%% AND times.file != "xml/factor2.xml"
%% WHERE times.job="treerepair" AND other.job="tt-classic";
\def\trpVsTTclassicAvg{10.5}
\def\trpVsTTclassicMed{9.3}

%% DEFMACRO REFORMAT(precision=1)
%% SELECT AVG(times.time / other.time) as trpVsTTrepairAvg,
%%     MEDIAN(times.time / other.time) as trpVsTTrepairMed
%% FROM time times JOIN time other ON times.file = other.file
%% AND times.file != "xml/factor2.xml"
%% WHERE times.job="treerepair" AND other.job="tt-repair";
\def\trpVsTTrepairAvg{6.2}
\def\trpVsTTrepairMed{5.5}

%% DEFMACRO REFORMAT(col 0-5=(precision=1) col 6=(precision=0))
%% SELECT AVG(tr.time  / gzip.time) as avgSlowdownClassic,
%%     MEDIAN(tr.time  / gzip.time) as medSlowdownClassic,
%%        AVG(ttr.time / gzip.time) as avgSlowdownRepair,
%%     MEDIAN(ttr.time / gzip.time) as medSlowdownRepair,
%%        AVG(trp.time / gzip.time) as avgSlowdownTreeRePair,
%%     MEDIAN(trp.time / gzip.time) as medSlowdownTreeRePair,
%% 100*MEDIAN(tr.time  / gzip.time - 1) as medSlowdownClassicPercent
%% FROM time gzip
%% JOIN time tr  ON gzip.file=tr.file  AND tr.job='tt-classic'
%% JOIN time ttr ON gzip.file=ttr.file AND ttr.job='tt-repair'
%% JOIN time trp ON gzip.file=trp.file AND trp.job='treerepair'
%% WHERE gzip.job='gzip9' AND gzip.file != "xml/factor2.xml"
\def\avgSlowdownClassic{2.0}
\def\medSlowdownClassic{1.6}
\def\avgSlowdownRepair{3.3}
\def\medSlowdownRepair{2.6}
\def\avgSlowdownTreeRePair{19.6}
\def\medSlowdownTreeRePair{21.2}
\def\medSlowdownClassicPercent{62}

%% DEFMACRO REFORMAT(precision=1)
%% SELECT AVG(bzip.time / tr.time ) as avgBzipSlowdownClassic,
%%     MEDIAN(bzip.time / tr.time ) as medBzipSlowdownClassic,
%%        AVG(bzip.time / ttr.time) as avgBzipSlowdownRepair,
%%     MEDIAN(bzip.time / ttr.time) as medBzipSlowdownRepair,
%%        AVG(bzip.time / trp.time) as avgBzipSlowdownTreeRePair,
%%     MEDIAN(bzip.time / trp.time) as medBzipSlowdownTreeRePair
%% FROM time bzip
%% JOIN time tr  ON bzip.file=tr.file  AND tr.job='tt-classic'
%% JOIN time ttr ON bzip.file=ttr.file AND ttr.job='tt-repair'
%% JOIN time trp ON bzip.file=trp.file AND trp.job='treerepair'
%% WHERE bzip.job='bzip2' AND bzip.file != "xml/factor2.xml"
\def\avgBzipSlowdownClassic{15.4}
\def\medBzipSlowdownClassic{11.9}
\def\avgBzipSlowdownRepair{9.7}
\def\medBzipSlowdownRepair{7.3}
\def\avgBzipSlowdownTreeRePair{1.8}
\def\medBzipSlowdownTreeRePair{1.6}

\paragraph*{Speed.}
Using our RePair-inspired combiner increases the running time of the top tree
creation stage, doubling it on average. Our implementation of classical top tree
compression was~\trpVsTTclassicAvg\xspace times faster than TreeRePair on
average over the corpus from Table~\ref{tbl:corpus}, and
still~\trpVsTTrepairAvg\xspace times faster when using our RePair combiner.
Detailed running time measurements are given in Table~\ref{tbl:runtime}. In
particular, classical top tree compression takes only twice as long as
\texttt{gzip -9} on average, and~\avgSlowdownRepair\xspace times when using our
RePair combiner (TreeRePair:~\medSlowdownTreeRePair). In contrast,
\texttt{bzip2} is~\avgBzipSlowdownClassic\xspace times \emph{slower} than top
tree compression on average, and~\avgBzipSlowdownRepair\xspace times when using
our RePair combiner. This strikingly demonstrates the method's qualities as a
fast compressor.

\paragraph*{Performance on Random Trees.}

% IMPORT-DATA random eval/random.txt
\begin{figure}[bt]
\center
\mbox {
\begin{tikzpicture}
  \begin{semilogxaxis}[
    height=7cm,
    width=12cm,
    grid=major,
    log basis x={2},
    xlabel={Tree size ($\NumN$)},
    ylabel={Compression Ratio / $\log_\NumS{\NumN}$},
    yticklabel style={/pgf/number format/fixed, /pgf/number format/precision=4},
  ]

  % instances with treesize >= 16777216 are averaged over 100 trees only
  % but the difference between 10 and 100 trees is already tiny.
  % from 134217728, 40 instances were used.

  %% PLOT SELECT treesize as x, (treesize * 1.0 / dagsize) / log(2, treesize) as y
  %% FROM random WHERE treesize >= 1024 ORDER BY x
  \addplot[mark=o,smooth,blue] coordinates { (1024,0.0899824) (1448,0.0894337) (2048,0.0892103) (2896,0.0892063) (4096,0.0892608) (5792,0.0894183) (8192,0.0896378) (11585,0.0898775) (16384,0.0900566) (23170,0.0900904) (32768,0.0900356) (46340,0.0898153) (65536,0.0895653) (92681,0.0893252) (131072,0.0891322) (185363,0.0890166) (262144,0.0889394) (370727,0.0889043) (524288,0.0889151) (741455,0.0889766) (1048576,0.0890796) (1482910,0.0892214) (2097152,0.0893796) (4194304,0.0897133) (8388608,0.0900301) (16777216,0.090251) (33554432,0.0903301) (67108864,0.0902785) (1.34218e+08,0.090144) (2.68435e+08,0.0899709) (5.36871e+08,0.0898073) };

  \end{semilogxaxis}
\end{tikzpicture}
}

\caption{Compression ratio in terms of the average number of edges, divided by
information-theoretical compression ratio bound~$\log_\sigma{\NumN}$,
over~1000 random trees of size~$\NumN=2^{10}$ to~$2^{29}$ with~$\NumS=2$. All
values are contained in a small range, suggesting asymptotically
optimal---$\Oh{\NumN/\log_\NumS{ \NumN}}$---compression.}\label{fig:exp:rand}
\end{figure}
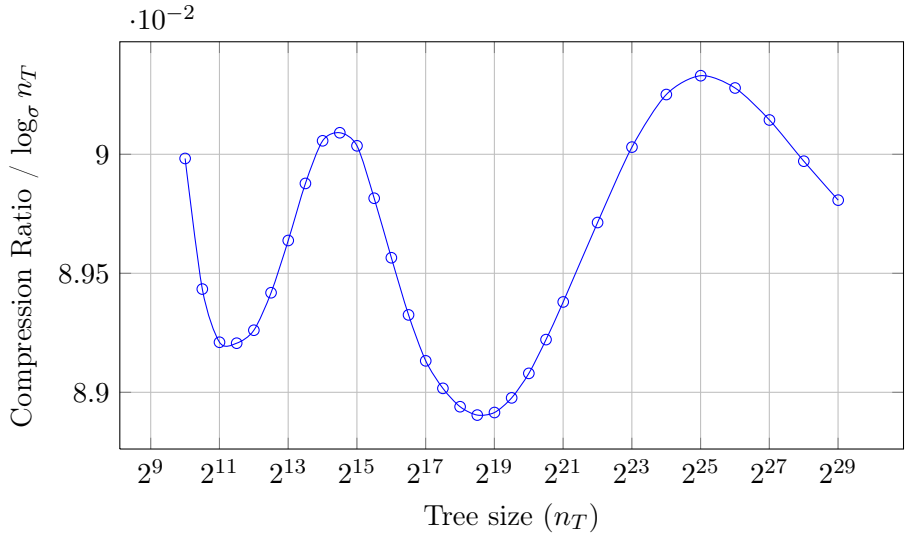

We examine random trees to show that tree compression with top trees is a very
versatile method, and that it does not rely on typical characteristics of XML
files. By definition, random trees do not compress well. Thus, we can use them
to approximate worst-case behaviour.  We generate trees uniformly at random
using a method developed by Atkinson and Sack~\cite{RandomTrees1992}---note
that the method is not limited to binary trees. For the generated trees, we compare the average
number of edges~$\NumNTD$ in the Top DAG to the information-theoretic lower
bound of~$\Omega(\NumN/\log_\sigma{\NumN})$ for a tree of size~$\NumN$. The
results of this are shown in Figure~\ref{fig:exp:rand} for~$\sigma=2$. We can
see that apart from some oscillation, the values are in a very small range
between~$0.0889$ and~$0.0904$ and do not show an overall tendency to grow or
shrink, except for the amplitude of oscillation. This suggests that tree
compression with top trees performs asymptotically optimal on random trees.

The oscillation or zig-zag behaviour exhibited in Figure~\ref{fig:exp:rand}
poses a riddle. The period duration doubles with each quarter of a period,
exhibiting exponential growth, while the wave's amplitude appears to grow by a
constant amount per quarter period. We do not have a definitive explanation
for the causes of this behaviour. However, we can speculate about possible
contributing factors. For one, consider the subtrees that could be shareable
in the top tree. Their height, and therefore number, grows logarithmically
with the height of the top tree, which in turn grows logarithmically with
respect to the input tree's size. Thus, the number of potentially shareable
subtrees grows proportional to~$\log{\log{\NumN}}$. As the potential for DAG
compression grows with the number of shareable subtrees, we would expect a
sawtooth-like pattern in the compression ratios, spiking whenever the
shareable subtree height increases.
This could contribute to the exponential growth in period duration. Further
investigation beyond the scope of this paper would be required to account for
the smoothness and amplitude of the curve.

%=========================================================================
%  Conclusion
%=========================================================================
\section{Conclusions}\label{s:concl}

We have demonstrated that tree compression with top trees is viable, and
suggested several enhancements to improve the degree of compression achieved.
Using the notion of combiners, we demonstrated that significant improvements can
be obtained by carefully choosing the order in which clusters are merged during
top tree creation. We showed that the worst-case compression ratio is
within a~$\log\log_\NumS{\NumN}$ factor of the information-theoretical bound,
and experiments with random trees suggest that actual behaviour is asymptotically
optimal.
Further, we gave efficient methods to navigate the compressed representation,
and described how the top DAG can be encoded to support efficient navigation
without prior decompression.

We thus conclude that tree compression with top trees is a very promising
compressor for labelled trees, and has several key advantages over other
compressors that make it worth pursuing. It is our belief that its great
flexibility, efficient navigation, high speed, simplicity, and provable bounds
should not be discarded easily. While further careful optimizations are required
to close the compression ratio gap, tree compression with top trees is already a
good and fast compressor with many advantages.

\paragraph*{Future Work} We expect that significant potential for improvement
lies in more sophisticated combiners. The requirements for combiners give us a
lot of space to devise better merging algorithms. Combiners might also be used
to improve locality in the top tree in addition to compression performance,
leading to better navigation performance. Moreover, additional compression
improvements should be achievable with carefully engineered output
representations. Since the vast majority of total running time is currently
spent on the construction of the top DAG, using more advanced encodings may
improve compression without losing speed. One starting point to replace our
relatively naïve representation could be a decomposition of the top DAG into two
spanning trees~\cite{ESP2013}.

\begin{table}[b!]
\caption{Running times in seconds, median over ten iterations}\label{tbl:runtime}
\footnotesize
\centerline{
\begin{tabular}{l rrrrrrr}
\toprule
\textbf{File name}~~ & ~~\textbf{TopTrees} & ~~\textbf{TT+RePair} & ~~\textbf{TreeRePair}
& ~~\textbf{RePair} & ~~\textbf{gzip -9} & ~~\textbf{bzip2} \\\midrule
%% TABULAR REFORMAT(precision=2)
%% SELECT SUBSTR(REPLACE(SUBSTR(SUBSTR(time.file, INSTR(time.file, "/")+1),-4,-99),"_","-"), 1, 15) as x,
%% (SELECT MEDIAN(time) from time t1 WHERE t1.file=time.file AND t1.job='tt-classic' GROUP BY t1.file),
%% (SELECT MEDIAN(time) from time t2 WHERE t2.file=time.file AND t2.job='tt-repair' GROUP BY t2.file),
%% (SELECT MEDIAN(time) from time t3 WHERE t3.file=time.file AND t3.job='treerepair' GROUP BY t3.file),
%% (SELECT MEDIAN(time) from time t4 WHERE t4.file=time.file AND t4.job='repair' GROUP BY t4.file),
%% (SELECT MEDIAN(time) from time t5 WHERE t5.file=time.file AND t5.job='gzip9' GROUP BY t5.file),
%% (SELECT MEDIAN(time) from time t6 WHERE t6.file=time.file AND t6.job='bzip2' GROUP BY t6.file)
%% FROM time
%% WHERE x != "factor2"
%% GROUP BY time.file ORDER BY lower(time.file)
 1998statistics & 0.00 &  0.01 &   0.05 &  0.04 & 0.00 &  0.13 \\
           dblp & 6.00 & 11.21 &  45.72 & 39.57 & 2.46 & 74.77 \\
enwiki-latest-p & 3.92 &  7.14 &  32.98 & 28.33 & 1.29 & 49.12 \\
       factor12 & 7.16 & 11.54 & 109.47 & 54.19 & 4.48 & 81.86 \\
      factor4.8 & 2.82 &  4.70 &  46.22 & 21.61 & 1.79 & 33.09 \\
        factor4 & 2.40 &  3.92 &  39.47 & 17.75 & 1.49 & 28.21 \\
        factor7 & 4.21 &  6.84 &  67.83 & 31.55 & 2.61 & 48.60 \\
  JST-gene.chr1 & 0.04 &  0.06 &   0.38 &  0.54 & 0.03 &  1.27 \\
   JST-snp.chr1 & 0.24 &  0.38 &   2.12 &  3.40 & 0.17 &  6.33 \\
           nasa & 0.15 &  0.23 &   0.94 &  0.85 & 0.06 &  1.86 \\
 NCBI-gene.chr1 & 0.31 &  0.51 &   2.25 &  3.33 & 0.20 &  7.97 \\
       proteins & 6.92 & 11.88 &  50.17 & 53.27 & 2.41 & 81.92 \\
      SwissProt & 1.13 &  2.13 &  12.35 &  5.74 & 0.50 & 11.15 \\
     treebank-e & 1.35 &  1.99 &  12.70 &  4.00 & 2.80 &  3.62 \\
            uwm & 0.01 &  0.02 &   0.11 &  0.09 & 0.00 &  0.28 \\
           wiki & 0.78 &  1.17 &   5.59 &  4.28 & 0.22 &  9.21 \\
% END TABULAR SELECT SUBSTR(REPLACE(SUBSTR(SUBSTR(time.file, INSTR(time.file,...
\bottomrule
\end{tabular}
}
\end{table}

% IMPORT-DATA gp eval/gp.txt
\begin{table}[b!]
\caption{Compressed file sizes in Bytes}\label{tbl:size}
\footnotesize
\centerline{
\begin{tabular}{l rrrrrrrr}
\toprule
\textbf{File name}~~ & ~~\textbf{Succinct} & ~~\textbf{TopTrees} & ~~\textbf{TT+RePair} & ~~\textbf{TreeRePair}
& ~~~~\textbf{RePair} & ~~~~\textbf{gzip -9} & ~~~~~~\textbf{bzip2}\\\midrule
%% TABULAR REFORMAT(precision=0 group=\,)
%% SELECT SUBSTR(REPLACE(SUBSTR(SUBSTR(time.file, INSTR(time.file, "/")+1),-4,-99),"_","-"), 1, 15) as x,
%% (SELECT (succinct+7)/8 from eval t0 WHERE t0.file=time.file AND repair=0 GROUP BY t0.file),
%% (SELECT (compressed+7)/8 from eval t1 WHERE t1.file=time.file AND repair=0 GROUP BY t1.file),
%% (SELECT (compressed+7)/8 from eval t2 WHERE t2.file=time.file AND repair=1 AND minratio=1.26 GROUP BY t2.file),
%% (SELECT size_bytes from treerepair t3 WHERE t3.file=time.file GROUP BY t3.file),
%% (SELECT (compressed+7)/8 from repair t4 WHERE t4.file=time.file GROUP BY t4.file),
%% (SELECT size from gp t5 WHERE t5.file=time.file AND t5.job='gzip9' GROUP BY t5.file),
%% (SELECT size from gp t6 WHERE t6.file=time.file AND t6.job='bzip' GROUP BY t6.file)
%% FROM time
%% WHERE x != "factor2"
%% GROUP BY time.file ORDER BY lower(time.file)
 1998statistics &      18\,426 &         788 &         851 &         692 &      1\,327 &      4\,080 &      1\,301 \\
           dblp & 13\,740\,160 & 1\,486\,208 & 1\,416\,538 & 1\,093\,533 & 2\,037\,878 & 2\,476\,347 & 1\,116\,311 \\
enwiki-latest-p &  7\,901\,904 &    516\,638 &    525\,532 &    379\,410 &    866\,161 & 1\,490\,278 &    544\,606 \\
       factor12 & 16\,402\,888 & 2\,069\,437 &    948\,167 &    705\,740 & 3\,092\,194 & 6\,342\,947 & 2\,913\,894 \\
      factor4.8 &  6\,565\,499 & 1\,070\,045 &    784\,519 &    548\,853 & 1\,587\,043 & 2\,542\,773 & 1\,168\,654 \\
        factor4 &  5\,473\,158 &    937\,660 &    704\,105 &    490\,945 & 1\,429\,872 & 2\,119\,269 &    973\,463 \\
        factor7 &  9\,571\,503 & 1\,421\,376 &    855\,063 &    625\,094 & 2\,248\,370 & 3\,702\,132 & 1\,700\,043 \\
  JST-gene.chr1 &      96\,159 &      5\,332 &      5\,523 &      3\,672 &      8\,273 &     33\,316 &      7\,027 \\
   JST-snp.chr1 &     547\,594 &     29\,084 &     27\,039 &     20\,654 &     50\,347 &    194\,862 &     49\,857 \\
           nasa &     341\,161 &     42\,077 &     39\,883 &     29\,310 &     60\,394 &     83\,231 &     34\,404 \\
 NCBI-gene.chr1 &     721\,803 &     17\,880 &     17\,418 &     11\,459 &     29\,912 &    199\,308 &     47\,901 \\
       proteins & 17\,315\,832 &    905\,613 &    860\,366 &    614\,892 & 1\,537\,249 & 3\,214\,663 & 1\,141\,697 \\
      SwissProt &  2\,343\,730 &    598\,960 &    574\,466 &    395\,417 &    699\,757 &    829\,119 &    398\,197 \\
     treebank-e &  2\,396\,061 & 1\,173\,463 & 1\,170\,304 &    830\,324 & 1\,537\,334 & 1\,858\,722 & 1\,032\,303 \\
            uwm &      37\,491 &      2\,177 &      2\,070 &      1\,366 &      3\,101 &      7\,539 &      2\,102 \\
           wiki &  1\,242\,418 &    110\,686 &    102\,371 &     75\,090 &    171\,075 &    247\,898 &     93\,858 \\
% END TABULAR SELECT SUBSTR(REPLACE(SUBSTR(SUBSTR(time.file, INSTR(time.file,...
\bottomrule
\end{tabular}
}
\end{table}

%=========================================================================
%  Bibliography
%=========================================================================

% When using Bibtex, the following form may be used.

\bibliographystyle{plain}
\bibliography{diss.bib}
\end{document}